\documentclass[12pt]{amsart}
\usepackage[margin=1in]{geometry}
\usepackage{mathptmx}
\usepackage{mathtools}
\usepackage{float}
\usepackage{amsthm}
\usepackage{amssymb}
\usepackage[alphabetic]{amsrefs}
\usepackage{graphicx}
\usepackage{tikz}
	\usetikzlibrary{decorations.pathreplacing}
	\usetikzlibrary{decorations.pathmorphing}
	\usetikzlibrary{patterns}
\usepackage{caption}

\newcommand{\ra}{\rightarrow}
\newcommand{\pr}{\prime}
\newcommand{\de}{\partial}
\newcommand{\te}{\theta}

\newcommand{\C}{\mathbb{C}}

\newcommand{\Q}{\mathbb{Q}}
\newcommand{\R}{\mathbb{R}}
\newcommand{\Z}{\mathbb{Z}}

\newcommand{\abs}[1]{\left\lvert #1 \right\rvert}

\newcommand{\lbar}[1]{\overline{#1}}
\newcommand{\id}{\mathrm{id}}

\newtheorem{thm}{Theorem}

\newtheorem{lemma}{Lemma}

\theoremstyle{definition}
\newtheorem{defin}{Definition}
\newtheorem*{defin*}{Definition}
\theoremstyle{remark}
\newtheorem{example}{Example}
\newtheorem{innercustomex}{Example}
\newenvironment{customex}[1]
	{\renewcommand\theinnercustomex{#1}\innercustomex}
	{\endinnercustomex}

\begin{document}

\title{Entanglement of Sections, Examples Looking for a Theory}

\author{Michael H. Freedman}
\address{\hskip-\parindent
	Michael H. Freedman \\
	Microsoft Research, Station Q, and Department of Mathematics \\
	University of California, Santa Barbara \\
	Santa Barbara, CA 93106
}

\author{Matthew B. Hastings}
\address{\hskip-\parindent
	Matthew B. Hastings \\
	Microsoft Research, Station Q \\
	Santa Barbara, CA 93106
}

\begin{abstract}
    Quantum information is about the entanglement of states. To this starting point we add parameters whereby a single state becomes a non-vanishing section of a bundle. We consider through examples the possible entanglement patterns of sections.
\end{abstract}

\maketitle

\section{Examples: Entanglement Constraints on Sections of Tensor Product Bundles}
Consider the elementary subject of finite dimensional linear algebra, perhaps with the vector spaces dressed with an Hermitian inner product, so as to become finite dimensional Hilbert spaces. It has two well known enhancements. If we add parameters we get bundle theory (in the dressed case the structure group reduces from GL$(n,\C)$ to U$(n)$). If we add fixed tensor structures to the Hilbert spaces---so as to have pieces under the respective control of Alice, Bob, Claire, etc., we get quantum information theory. This note advertises the ``push out'' of these two enhancements. Should we be doing quantum information in families? Or when we meet a tensor product bundle, should we ask not just about its subbundles and (non-vanishing) sections, but about their entanglement properties as well?

This note is inventing, and asking its readers to help invent the upper right corner of this ``push out.''

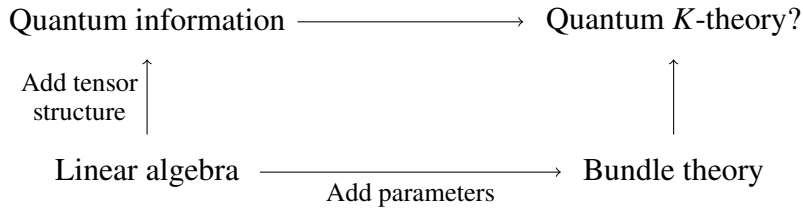
\begin{figure}[htp]
    \centering
    \begin{tikzpicture}
        \node at (0,0) {Quantum information};
        \draw[->] (2,0) -- (5,0);
        \node at (7,0) {Quantum $K$-theory?};
        \node at (0,-2) {Linear algebra};
        \draw[->] (1.5,-2) -- (5.5,-2);
        \node at (3.5,-2.3) {\footnotesize{Add parameters}};
        \node at (7,-2) {Bundle theory};
        \draw[->] (0,-1.5) -- (0,-0.5);
        \node at (-0.9,-0.8) {\footnotesize{Add tensor}};
        \node at (-0.9,-1.2) {\footnotesize{structure}};
        \draw[->] (7,-1.5) -- (7,-0.5);
    \end{tikzpicture}
\end{figure}

Let us mention two contexts where the quantum world is less homogeneous than the classical. First, if $\mathcal{H} \cong H_1 \otimes \cdots \otimes H_k$ the unit vectors, ``states'' of $\mathcal{H}$ do not all look the same. The symmetry group $\operatorname{U}(d_1) \times \cdots \times \operatorname{U}(d_k)$, $d_i = \dim(\mathcal{H}_i)$, is much smaller than U$(d)$, $d = \prod_{i=1}^k d_i$. Different directions have vastly different entanglement properties. Second, consider, for simplicity, the case where all $d_i = 2$, so $d = 2^k$. The Lie algebra of U$(2^k)$ is spanned by $i$ times the ``Pauli words'': $\{\sigma_i \otimes \cdots \otimes \sigma_{i_k}\}$ where $i_j = 0,1,2,$ or 3, $1 \leq j \leq k$, and $\sigma_0 = \id$, $\sigma_1 = X = \begin{vmatrix}
    0 & 1 \\ 1 & 0
\end{vmatrix}$, $\sigma_2 = Y = \begin{vmatrix}
    0 & -1 \\ i & 0
\end{vmatrix}$, and $\sigma_3 = \begin{vmatrix}
    1 & 0 \\ 0 & -1
\end{vmatrix}$. Except for the identity, the $4^k-1$ remaining Pauli words which span the simple Lie algebra SU$(2^k)$ are all iso-spectral, and thus mutually equivalent under conjugation, i.e.\ the adjoint representation. However, from a quantum perspective, they are very different. The number $w$, called the weight, $1 \leq w \leq k$, of non-identity Pauli matrices in the Pauli word tells us how many, $w$, bodies are coupled by that Hamiltonian. So in a quantum treatment is has been proposed \cite{bs18} that one should study left invariant metrics, e.g.\ $g_{ij} = e^{w_i} \delta_{ij}$, diagonal in the Pauli-word basis, which reflect this distinction, rather than the usual Killing form $\langle H_1, H_2 \rangle :: \operatorname{tr}(\operatorname{ad} H_1 \circ \operatorname{ad} H_2)$, which does not. These inhomogenieties, so important in quantum information, should not be lost sight of simply because one has a problem with continuous parameters, but instead should be accounted for in a theory of tensor bundles.

We can think of two ways, there may be others, in which quantum problems come with continuous parameters. First, in periodic systems, the dof in the unit cell constitute the fiber of a bundle over the momentum torus, or Brillion zone. This point of view was important in understanding the ``10-fold way'' a topological classification of free fermion states \cites{kit09,hastings13}. Second, one might consider the configuration space of experimental conditions for a system $F_p$, with parameter $p$. The effective configuration space is one where gauge equivalent parameter settings (e.g.\ multiplies of a flux quantum) are identified, and so could have a rich topology.

These are our motivations for studying the upper-right corner. What we present next are ad hoc, but perhaps exemplary, calculations, useful to initiate this discussion. What made $K$-theory a powerful tool 60 years ago were the \emph{regularities}, such as Bott periodicitity, that classical calculations uncovered. As yet, in the quantum case, we see no over-riding regular patterns, but propose searching for them. This section is organized around five examples, which can serve as a starting point.

We give examples of tensor products of vector bundles where the factor bundles $A$ and $B$ (also $A$, $B$, and $C$ in later examples) have no nonvanishing sections, yet the tensor product $P^\pr = A \otimes B$ (and $P = A \otimes B \otimes C$) does admit a nonvanishing section. So far, this is no surprise as tensor products multiply fiber dimension, and when the real dimension of the fiber exceeds that of the base, there will always be a nonvanishing section. What is interesting is that much can be said about the patterns of entanglement present in such nonvanishing sections. We give examples of $P^\pr$ and $P$ as above where any nonvanishing section $\Gamma$ must at some point $x \in X$ of the base have nongeneric entanglement. The constraints on entanglement may be thought of as ``quantum characteristic classes'' although it is not cohomological in the usual sense. This interpretation will be discussed later. MF would like to thank Peter Teichner for insightful discussions.

As a warm up, consider two nontrivial (complex) vector bundles $A$ and $B$ over a base $X$. Further suppose that $A$ and $B$ admit no nonvanishing section. Now consider $P^\pr = A \otimes B$. To stick to the most familiar setting, let us assume that all bundles have unitary structure group $\operatorname{U}(n)$, when the fiber is $\C^n$. This can be easily relaxed to $\operatorname{GL}(n,\C)$, but there is no present need. The fibers of $P^\pr$, isomorphic to $\C^a \otimes \C^b$ (here $a$ and $b$ are the fiber dimensions of $A$ and $B$, respectively) are not canonically identified with $\C^a \otimes \C^b$, but the identification \emph{is} canonical up to the product action of $\operatorname{U}(a) \times \operatorname{U}(b)\slash \operatorname{diag} \mathrm{U}(1)$ on the respective factors. This enables us to speak unambiguously about the degree of entanglement of a nonvanishing section $\Gamma(x)$ of $P^\pr$ at every point $x \in X$ in the base, since by definition measures of entanglement are unchanged by local unitary (LU) operators.

\begin{example}
    We use the 2-sphere $S^2 \cong \C P^1$ as our base space. Let $A$ be the complex line bundle with $c_1(A) = 1 \in H^2(S^2; \Z)$ and $B$ be the complex conjugate of $A$, $B = \lbar{A}$, so $c_1(B) = - 1 \in H^2(S^2;\Z)$. The line bundle $A \otimes B \cong \operatorname{Hom}(A,A)$ is canonically trivial since the structure group $\operatorname{U}(1)$ is abelian. More concretely, it is well-defined to rotate any fiber of $A$ by multiplying by $e^{2\pi i \te}$. Thus $A \otimes B$ has a nonvanishing section: the identity rotation over each point $x \in S^2$. Since the fiber dimensions are one, it is not meaningful to ask about entanglement of the sections: there can be none.
\end{example}

All sections, discussed in these notes, are assumed to be nowhere vanishing. The preceeding example motivates the definition below.

Suppose a bundle has a fixed tensor factorization $P^\pr = A \otimes B$ over base $X$ (or $P = A \otimes B \otimes C$, etc.). We say a section is \emph{simple} (i.e.\ tensor-rank 1) if it has the form $\Gamma(x) = \vert a \rangle \vert b \rangle$ (or $\Gamma(x) = \vert a \rangle \vert b \rangle \vert c \rangle$, etc.). Any simple section induces projective sections $\vert \hat{a}(x) \rangle \coloneqq \hat{\Gamma}_A(x) \in \hat{A}$, $\vert \hat{b}(x) \rangle \coloneqq \hat{\Gamma}_B(x) \in \hat{B}$, etc., where $\hat{A}$ ($\hat{B}$, ...) is obtained by dividing all unit fibers by phase, yielding new $\C P^\ast$-bundles. For each projective bundle $\hat{A}$ ($\hat{B}$, ...) there is a first Chern class $c_{\hat{a}}$ ($c_{\hat{b}}$, ...), the first Chern class of the $\operatorname{U}(1)$ bundle of phases that was divided out in passing from $A$ to $\hat{A}$ ($B$ to $\hat{B}$, ...) which is the unique obstruction lying in $H^2(X; \Z)$ to lifting $\vert \hat{a} \rangle$ ($\vert \hat{b} \rangle$, ...) to a section $\vert a \rangle$ ($\vert b \rangle$, ...) of $A$ ($B$, ...). Since $\vert a \rangle \vert b \rangle$ was originally a section of $A \otimes B$, the obstruction sums to zero:
\begin{equation}
    c_{\hat{a}} + c_{\hat{b}} = 0
\end{equation}

\begin{defin}\label{def:untwisted}
    If the individual Chern class obstructions $c_A$ and $c_B$ both vanish, we say $\Gamma(x)$ is \emph{untwisted}. When $P = A \otimes B \otimes C$, a section of $\Gamma(x)$ will be called \emph{untwisted} iff $\Gamma(x) \in V_2^\pr$, and where $\Gamma(x) \in \lbar{A - BC}$ ($\lbar{B - CA}$ or $\lbar{C - AB}$), bar denoting closure, and $c_{\hat{a}} = 0 = c_{\hat{b}\hat{c}}$ ($c_{\hat{b}} = 0 = c_{\hat{c}\hat{a}}$, $c_{\hat{c}} = 0 = c_{\hat{a}\hat{b}}$). It follows from additivity (line 1) that where $\Gamma(x) \in A - B - C$, $c_{\hat{a}} = c_{\hat{b}} = c_{\hat{c}} = 0$.
\end{defin}

In Definition \ref{def:untwisted} we have used the notation $A - B - C$ for vectors (sections) that are simple 3-fold products, $A - BC$ for the section of the form $\vert a \rangle \otimes \vert \phi_{BC}\rangle$, i.e.\ simple w.r.t.\ the first factor and the last two combined, etc. $V_2^\pr$, the strata of at most bipartite entanglement, is defined as $A-B-C \cup A-BC \cup B-AC \cup C - AB$.

\begin{example}\label{ex:s4}
    Let $S^4 \cong QP^1$ be the base and let $A$ be the complex 2-plane bundle with $c_2(A) = 1 \in H^4(S^4;\Z)$ and again $B = \lbar{A}$. Now $P^\pr = A \otimes B \cong \operatorname{Hom}(A,A)$ is nontrivial. This may be checked by computing its Chern character to be $1 + 2(\text{generator}) \text{ of } H^4$. The Chern character is a ring isomoprhism from $KU \otimes Q \ra \oplus_{i \text{ even}} H^i(;\Q)$ from complex $K$-theory tensor $Q$ to the even rational cohomology. Since the fiber of $P^\pr$ has 8 real dimension, larger than the dimension 4 of the base, $P^\pr$ certainly has a nonvanishing section. We can ask: Can such a section $\Gamma$ be unentangled? That is, can $\Gamma(x)$ be a tensor-rank 1 vector for all $x \in S^4$? Let's imagine writing $\Gamma(x) = v_x \otimes w_x$, where $v \in A_x$ and $w \in B_x$, the respective fibers. The only ambiguity in the choice of $v_x$ and $w_x$ is a phase; for any $\te$ we may alter the choice: $v_x \mapsto e^{2\pi i \te} v_x$, $w_x \mapsto e^{-2\pi i \te}w_x$. But letting $\hat{v}_x$ and $\hat{w}_x$ lie in the projective spaces $(A_x \setminus 0)\slash \text{phase}$ and $(B_x \setminus 0) \slash \text{phase}$, there is no longer any ambiguity. Thus, a supposed tensor-rank 1 section $\Gamma(x)$ would induce section of the projectivized bundles $\hat{A}$ and $\hat{B}$:

    \begin{figure}[ht]
        \centering
        \begin{tikzpicture}[baseline=-5.8ex]
            \node at (0,0) {$\C P^3$};
            \draw [->] (0.45,0) -- (0.85,0);
            \node at (1.1,0) {$\hat{A}$};
            \draw [->] (1.1,-0.3) -- (1.1,-0.75);
            \node at (1.1,-1) {$S^4$};
            \draw[->] (1.4,-0.9) to[out=60,in=-60] (1.4,-0.1);
            \node at (1.75,-0.5) {\footnotesize{$\hat{v}_x$}};

            \node at (2.4,-1.2) {,};

            \node at (3.2,0) {$\C P^3$};
            \draw [->] (3.65,0) -- (4.05,0);
            \node at (4.3,0) {$\hat{B}$};
            \draw [->] (4.3,-0.3) -- (4.3,-0.75);
            \node at (4.3,-1) {$S^4$};
            \draw[->] (4.6,-0.9) to[out=60,in=-60] (4.6,-0.1);
            \node at (5,-0.5) {\footnotesize{$\hat{w}_x$}};
        \end{tikzpicture}
    \end{figure}

    But this cannot be; the obstruction to lifting a section of $\hat{A}$ ($\hat{B}$) to a section of $A$ ($B$) is an element $c_{\hat{a}}\ (c_{\hat{b}}) \in H^2(S^4;\Z) \cong 0$, and hence vanishes. A rank one section $\Gamma(x)$ of $A \otimes B$ would lead to section(s) of $A$ ($B$), a contradiction. So while $A \otimes B$ admits a nonvanishing section, any such must have tensor-rank $>1$ at some points of the base, i.e., be entangled.
\end{example}

\begin{customex}{2\textprime}
    Let $A$ be the $\C^2$-bundle over $S^4 \cong QP^1$, the generalized Hopf bundle already considered in Example \ref{ex:s4}. This time, consider the tensor square $A \otimes A$. By essentially the same reasoning as above, $A \otimes A$ contains no simple (tensor-rank 1) section. Note that $A \otimes A \cong \Lambda^2(A) \oplus \mathrm{sym}^2(A)$ decomposes into a direct sum of a 1D skew-symmetric (``singlet'') piece and a 3D symmetric piece. Since $H^2(S^4;\Z) \cong 0$, $\Lambda^2(A)$ has a non-vanishing entangled section, of the form $\vert0\rangle \vert1\rangle - \vert 1 \rangle \vert 0 \rangle$ at each point.
\end{customex}

Now let us turn to the question of which entanglement patterns can be seen in an arbitrary section of $\mathrm{sym}^2(A)$. For two qubits the most natural invariant, up to local unitary transformations (LU), is the von Neumann entropy $S(\psi)$. To choose a normalized form, $\psi = \cos \te \vert 0 \rangle \vert 0 \rangle + \sin \te \vert 1 \rangle \vert 1 \rangle$, one checks that $S(\psi) = \cos \te \log(\cos \te) + \sin \te \log(\sin \te)$. Note that in considering symmetric states $\psi$ up to LU-equivalene, we intentionally break this symmetry by allowing U$(2) \times \operatorname{U}(2) \slash \operatorname{U}(1)$ to act, that is, U$(2)$ acts independently on the two factors. By the argument of Example \ref{ex:s4}, which used $H^2(S^4;\Z) \cong 0$, any non-vanishing section $\psi$ of $\mathrm{sym}^2(A)$ must have an entangled vector $\psi(s)$ for some $s \in S^4$ which we call the ``south pole.'' Under LU-equivalence, according to the Schmidt decomposition, $\psi(s)$ will assume a form: $\psi(s) = \vert 0 \rangle \vert 0 \rangle + t \vert 1 \rangle \vert 1 \rangle$, $0 < t \leq 1$, up to an overall real normalization $\frac{1}{\sqrt{1 + t^2}}$, which we drop. It is now convenient to projectivize all bundles and sections, so we write:
\begin{equation}
    \begin{tikzpicture}[baseline=(current bounding box.center)]
        \node at (0,0) {$\mathrm{CP}^2 \ra P(\mathrm{sym}^2(A))$};
        \draw[->] (0.5,-0.35) -- (0.5,-0.8);
        \node at (0.5,-1.2) {$s \subset S^4$};
        \draw[->] (-0.2,-1) to[out=120,in=240] (-0.2,-0.3);
        \node at (-0.8,-0.65) {\footnotesize{$P(\psi)$}};
    \end{tikzpicture}
\end{equation}

This projective bundle is of course trivial over the contractible northern patch $S^4 \setminus s$, so we may write $P(\psi)$ as a function, which, abusing notation slightly, we still call $P(\psi): (D^4,S^3) \ra \C P^2$, where $S^3$ is the completed ``infinity'' of $S^4 \setminus s$. The boundary values of $P(\psi)$ on $S^3$ are determined by the symmetric square of the clutching function $c$ for $A$. $c$ may be taken to be any orientation preserving diffeomorphism: $S^3 \xrightarrow{c} \operatorname{SU}(2)$, recalling that $A$ is the generalized Hopf bundle.

Next, we should describe some internal structure of $P(\operatorname{sym}^2(\C^2)) \cong P(\C^3) \cong \C P^2$, with its use in understanding the map $c$, above. First, consider the degree 2 curve, $C_2$, in $P(\operatorname{sym}^2(\C^2))$, consisting of \emph{unentangled} triplets of $\operatorname{sym}^2(\C^2)$. Topologically, $C_2$ is a 2-sphere lying degree 2 in $\C P^2$. 
We claim that the space of maximally entangled triplets, which we call $\mathrm{Max}$, is diffeomorphic to $\R P^2$, and that
the complement $\C P^2 \setminus C_2$ deformation retracts to $\mathrm{Max}$.
Likewise, the complement $\C P^2 \setminus \mathrm{Max}$ deformation retracts to $C_2$.
( In fact, the region in between $C_2$ and $\mathrm{Max}$ consist of a product family of $SU(2)$-orbits $SU(2)/Z_4,$ each diffeomorphic to the lens space $L_{4,1}$.)

To show that $\mathrm{Max}$ is $\R P^2$, given an arbitrary symmetric state, $a|00\rangle+b(|01\rangle+|10\rangle)+c|11\rangle$, define a matrix
$$M\equiv \begin{pmatrix} a & b \\ b & c \end{pmatrix}.$$
It will be convenient to normalize the state to have $\ell_2$ norm $\sqrt{2}$ so that $|a|^2+2|b|^2+|c|^2=2$.  Then,
the Hilbert-Schmidt norm of $M$ is $\sqrt{2}$.
With this normalization, the maximally entangled states are those for which $M$ is a unitary matrix, while the product states are those for which $M$ is singular.
Since the phase of the state is arbitrary, we may choose the phase such that $M$ has unit determinant, and thus is an element of $SU(2)$.
In this case, $b$ must be pure imaginary and $a=c^*$.
Thus, $M$ is defined by three real numbers (real and imaginary parts of $a$ and imaginary part of $b$), whose sum of squares is equal to $1$, so $M$ is defined by a point in $S^2$.  However, if we simultaneously change the sign of $a,b,c$, then this gives the same state up to phase (and this is the only remaining phase arbitrariness once we have imposed that $M$ has unit determinant), so we identify opposite points of $S^2$, giving us $\R P^2$.

Now we show that $\C P^2 \setminus C_2$ deformation retracts to $\mathrm{Max}$.  Let $f_s(x):\R^+\rightarrow \R^+$ be a family of functions defined by
$$f_s(x)\equiv \frac{1+s\sqrt{x}}{1+sx},$$
for $s\in[0,\infty)$ so that $f_0(x)=1$ while $f_s(x) \rightarrow 1/\sqrt{x}$ as $s\rightarrow\infty$.
Then the family of matrices
$$Z_s(M) f_s(M M^*) M$$ gives the desired deformation retraction, where $Z_s(M)$ is a positive real scalar chosen to keep the Hilbert-Schmidt norm constant.
If $M$ has singular value decomposition $M=U \Lambda V$, for $U,V$ unitary and $\Lambda$ diagonal and non-negative real, then
$f_s(M M^*)M=U (f(\Lambda^2) \Lambda) V$ and so for non-singular $M$, $M$ converges to a unitary matrix as $s\rightarrow \infty$.
Further, $f_s(M M^*)M$ is symmetric: $(f_s(M M^*)M)^T=M f_s(M^* M)=f_s(M M^*)M$, where the superscript $T$ denotes transposition, where the first equality uses that $M=M^T$, and the second equality can be checked using the singular value decomposition.

Finally we show that $\C P^2 \setminus \mathrm{Max}$ deformation retracts to $C_2$.  Let $g_s(x):R\rightarrow R^+$ be a family of functions
defined by
$$g_s(x)\equiv \exp(s x).$$
Then the family of matrices
$$Z'_s(M) g_s(M M^*) M$$ for $s\in [0,\infty)$ gives the desired deformation retraction, where $Z'_s(M)$ is a positive real scalar chosen to keep the Hilbert-Schmidt norm constant.  Indeed, so long as the singular values of $M$ are distinct from each other, then the ratio between the singular values of $g_s(M M^*) M$ diverges as $s\rightarrow \infty$, and so $Z'_s g_s(M M^*) M$ converges to a singular matrix corresponding to a state in $C_2$.
Further, $g_s(M M^*) M$ is symmetric, with the same proof as for $f_s(M M^*)M$.

In fact, $\C P^2$ is the union of $D^2$-bundles over $C_2$ and $\mathrm{Max}$, along their common boundaries.  Unless the maximum entanglement is already achieved over the south pole $s$, the map $d:S^3-\rightarrow \C P^2$ (using the coordinates determined by our clutching function $c$) that describes the behavior of $P(\Psi)$ as one approaches the south pole, $s$, i.e. at the boundary of $D^4$, must take its values in $\mathcal{N}(C_2) :=\C P^2\setminus \mathrm{Max}$. We have just seen that $\mathcal{N}(C_2)$ deformation retracts to $C_2 \cong S^2$, so the map $d$ may be denoted by its Hopf degree in the diagram below. $d$ has Hopf degree one, meaning it is a generator of $\pi_3(S^2)$. The proof is completed by considering this diagram:

\begin{equation}
    \begin{tikzpicture}[baseline=(current bounding box.center)]
        \node at (0,0) {$D^4 \xrightarrow{\hspace{1.3em} P(\Psi) \hspace{1.3em}} \C P^2 \cong P(\operatorname{sym}(\C^2))$};
    \node[rotate=90] at (0.5,-0.6) {$\lhook\joinrel\longrightarrow$};
    \node at (0.6,-1.3) {$\C P^2 \setminus \mathrm{Max}$};
    \node[rotate=90] at (0.5,-2) {$\lhook\joinrel\longrightarrow$};
    \node at (-0.8,-2.7) {$S^3 \xrightarrow{\hspace{1em} \text{Hopf deg.\ 1} \hspace{1em}} \mathcal{N}(C_2)$};
    \node[rotate=90] at (-2.8,-1.35) {$\lhook\joinrel\xrightarrow{\hspace{3em}}$};
    \draw[->,dashed] (-2.4,-0.4) -- (-0.5,-1.2);
    \draw (-1.7,-0.9) -- (-1.4,-0.6);
    \end{tikzpicture}
\end{equation}

The dashed arrow would contradict Hopf-degree one, so it does not exist. From this, we conclude the extension $P(\Psi)$ over $D^4$ must, at some point, take at least one value in $\mathrm{Max}$. Thus, every non-zero section of $\operatorname{sym}^2(A)$ contains at least one vector $P(\Psi)(x)$ of maximal entanglement in $\C^2 \otimes \C^2 \supset \operatorname{sym}^2(\C^2)$. \qed

We have seen how a bundle built from tensoring two sectionless bundles might have sections but these must obey entanglement constraints. Now let us turn our attention to triple tensor products $P = A \otimes B \otimes C$. In this context there is a beautiful hierarchy of entanglement under the rather coarse SLOCC\footnote{SLOCC stands for Stochastic Local Operators and Classical Communication.}-equivalence relation. We will recapitulate the basic low dimensional SLOCC classification \cite{dvc00} from an algebraic geometric point of view and use what we learn to explain and explore further examples.

Before coming to examples of bundle triple tensor products, we review the SLOCC equivalence relation, following \cite{dvc00} but with an emphasis on the algebraic geometry. Let $\psi$ and $\psi^\pr$ be nontrivial vectors in a finite tensor product of finite dimensional Hilbert spaces:
\begin{equation}
    \psi, \psi^\pr \in H_1 \otimes \cdots \otimes H_n
\end{equation}
The definition of SLOCC-equivalent that we use \cites{lp99,v99} is that $\psi \equiv \psi^\pr$ iff there exists invertible matrices $M_1, \dots, M_n$ such that:
\begin{equation}\label{eq:slocc}
    \psi^\pr = M_1 \otimes \cdots \otimes M_n(\psi)
\end{equation}
Similarly, we say $\psi \geq \psi^\pr$ if line \ref{eq:slocc} can be written dropping the invertability assumption.

In the case of two tensor factors, $n = 2$, it is easy to check that $\psi \equiv \psi^\pr$ iff they have equal Schmidt ranks. Now we turn to three factors, $n = 3$, and also assume each $H_i \cong \C^2$ is a qubit. In a slight abuse of notation we label the three qubits $A$, $B$, and $C$ and retain the same notation for the $\C^2$-bundles in which these qubits have become fibers parameterized over a base. Let $\rho_A, \rho_B, \rho_C, \rho_{AB}, \rho_{BC}$, and $\rho_{CA}$ be the reduced density matrices of $\psi$ where the projector $\abs{\psi \rangle \langle \psi}$ has had the complementary indexed factors traced out, e.g.\ $\rho_A$ is the image of $\abs{\psi \rangle \langle \psi} \in A \otimes B \otimes C \otimes A^\ast \otimes B^\ast \otimes C^\ast$ under
\begin{equation}
    A \otimes B \otimes C \otimes A^\ast \otimes B^\ast \otimes C^\ast \xrightarrow{\operatorname{tr}_{BC}} A \otimes A^\ast
\end{equation}

It turns out that the ranks $r(\rho_A) = r(\rho_{BC})$, $r(\rho_B) = r(\rho_{AC})$, $r(\rho_C) = r(\rho_{AB})$, and the number of pure (tensor-rank 1) vectors in the images (ranges) $R(\rho_{BC}), R(\rho_{AC})$, and $R(\rho_{AB})$ determine the SLOCC equivalence classes. Let us begin with the answer. In this setting (three qubits) there are 6 equivalence classes, with a canonical representative listed in the partial order (above means $\geq)$ below:

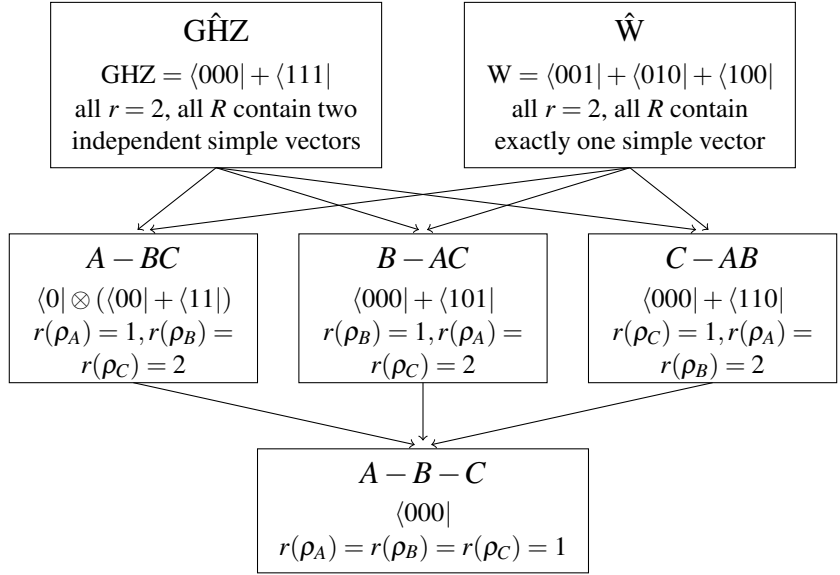
\begin{figure}[ht]
    \centering
    \begin{tikzpicture}[scale=1.1]
        \draw (-4.5,3) rectangle (-0.5,1);
        \node at (-2.5,2.7) {$\hat{\mathrm{GHZ}}$};
        \node at (-2.5,2.1) {\footnotesize{GHZ $= \langle 000\vert + \langle 111 \vert$}};
        \node at (-2.5,1.7) {\footnotesize{all $r =2$, all $R$ contain two}};
        \node at (-2.5,1.3) {\footnotesize{independent simple vectors}};
    
        \draw (0.5,3) rectangle (4.5,1);
        \node at (2.5,2.7) {$\hat{\mathrm{W}}$};
        \node at (2.5,2.1) {\footnotesize{W $= \langle 001\vert + \langle 010 \vert + \langle 100\vert$}};
        \node at (2.5,1.7) {\footnotesize{all $r =2$, all $R$ contain}};
        \node at (2.5,1.3) {\footnotesize{exactly one simple vector}};
    
        \draw (-5,0.2) rectangle (-2,-1.6);
        \node at (-3.5,-0.1) {$A - BC$};
        \node at (-3.5,-0.6) {\footnotesize{$\langle 0 \vert \otimes (\langle 00 \vert + \langle 11 \vert)$}};
        \node at (-3.5,-1) {\footnotesize{$r(\rho_A) = 1, r(\rho_B) =$}};
        \node at (-3.5,-1.4) {\footnotesize{$r(\rho_C) = 2$}};
    
        \draw (-1.5,0.2) rectangle (1.5,-1.6);
        \node at (0,-0.1) {$B - AC$};
        \node at (0,-0.6) {\footnotesize{$\langle 000 \vert + \langle 101 \vert$}};
        \node at (0,-1) {\footnotesize{$r(\rho_B) = 1, r(\rho_A) =$}};
        \node at (0,-1.4) {\footnotesize{$r(\rho_C) = 2$}};
    
        \draw (2,0.2) rectangle (5,-1.6);
        \node at (3.5,-0.1) {$C - AB$};
        \node at (3.5,-0.6) {\footnotesize{$\langle 000 \vert + \langle 110 \vert$}};
        \node at (3.5,-1) {\footnotesize{$r(\rho_C) = 1, r(\rho_A) =$}};
        \node at (3.5,-1.4) {\footnotesize{$r(\rho_B) = 2$}};
    
        \draw (-2,-2.4) rectangle (2,-3.9);
        \node at (0,-2.7) {$A - B - C$};
        \node at (0,-3.2) {\footnotesize{$\langle 000 \vert$}};
        \node at (0,-3.6) {\footnotesize{$r(\rho_A) = r(\rho_B) = r(\rho_C) = 1$}};

        \draw[<->] (-0.05,0.25) -- (-2.5,1) -- (-3.45,0.25);
        \draw[->] (-2.5,1) -- (3.3,0.25);
        \draw[<->] (0.05,0.25) -- (2.5,1) -- (3.45,0.25);
        \draw[->] (2.5,1) -- (-3.3,0.25);
        \draw[->] (3.5,-1.6) -- (0.1,-2.35);
        \draw[->] (0,-1.6) -- (0,-2.3);
        \draw[->] (-3.5,-1.6) -- (-0.1,-2.35);
    \end{tikzpicture}
    \caption{The 6 SLOCC classes in $\C^2 \otimes \C^2 \otimes \C^2$ with representatives.}
\end{figure}

In the lowest class, $\psi$ completely factorizes. In the second tranche of three classes, $\psi$ factors once. The interesting point is that there are \emph{two} distinct classes, which we write $\hat{\mathrm{GHZ}}$ and $\hat{\mathrm{W}}$ after their well-known representatives GHZ and W, respectively, in which all three qubits are entangled.

What separates $\hat{\mathrm{GHZ}}$ from $\hat{\mathrm{W}}$ is that in GHZ the three ranges $R(\rho_{BC}), R(\rho_{CA})$, and $R(\rho_{AB})$ contain two distinct tensor-rank one (simple) vectors (of $B \otimes C$, ...) whereas for $\psi \in \hat{\mathrm{W}}$ these ranges contain only a single ``double root'' simple vector. In order to understand this distinction, let us think geometrically. The reader may have already noticed that we do not bother to normalize the state vectors, but for the moment do. Then $\psi \in S^{15}$, the unit sphere in $(\C^2)^{\otimes 3} \cong \C^8$. It turns out that the union of the strata of increasing complexity are real algebraic varieties in $S^{15}$ with the indicated dimensions:
\begin{align*}
    & \dim_R(V_1^\pr) = \dim_R(A - B - C) = 7 \\
    & \dim_R(V_2^\pr) = \dim_R(A - B - C \cup A - BC \cup B - AC \cup C - AB) = 9 \\
    & \dim_R(V_3^\pr) = \dim_R(A - B - C \cup A - BC \cup B - AC \cup C - AB \cup \hat{\mathrm{W}}) = 13 \\
    & \dim_R(S^{15}) = \dim_R(A - B - C \cup A - BC \cup B - AC \cup C - AB \cup \hat{\mathrm{W}} \cup \hat{\mathrm{GHZ}}) = 15
\end{align*}

The first three real varieties are called $V_1^\pr, V_2^\pr,$ and $V_3^\pr \subset S^{15}$, respectively. We previously met $V_2^\pr$ in example \ref{ex:s4}. It is also useful to divide out by phase everywhere and obtain complex projective varieties $V_1, V_2,$ and $V_3 \subset \C P^7$ with $\dim_{\C}(V_1) = 3$, $\dim_{\C}(V_2) = 4$, and $\dim_{\C}(V_3) = 6$.

The condition, which we will shortly return to, regarding the count of simple vectors in the range of density matrices such as $\rho_{BC}$ can also be understood geometrically. $\rho_{BC}: \C_B^2 \otimes \C_C^2 \ra \C_B^2 \otimes \C_C^2 \cong \C^4$ is a linear map with range $R_{BC} \coloneqq R(\rho_{BC})$. A complex 2-plane, call it $R(\rho_{BC}) \coloneqq Q^\pr \subset \C^4$, defines a projective line $Q \subset \C P^3$. The condition that a vector in $\C_B^2 \otimes \C_C^2$ is simple, i.e. is of the form $\vert \psi_B \rangle \vert \psi_C \rangle$, is quadratic. The locus of all simple vectors is a nonsingular degree 2 hypersurface $S \subset \C P^3$, $S$ is birationally equivalent (and diffeomorphic to) the projective locus of $x^2 + y^2 + z^2 + w^2 = 0$; $S \cong \C P^1 \times \C P^1 \cong S^2 \times S^2$. Thus, projectively, the locus of simple vectors in $R_{BC}$ is $S \cap Q$. $Q$ is a degree 1 curve of $\C P^3$ and $S$ is a degree 2 hypersurface so the cohomology ring structure\footnote{$H^\ast(\C P^3;\Z)$ is a truncated polynomial algebra over $\Z$ with a single genator $x \in H^2(\C P^3;\Z)$ with relation $x^4 = 0$. $Q$, being a projective line, represents $x^2$ under Poincar\'{e} duality and $S$ represents $2x$ being degree 2. $x^2 \cup 2x = 2x^3$, which is Poincar\'{e} dual to two points.} of $\C P^3$ tells us that algebraically $[S] \cap [Q] = 2$. If the intersection is transverse, the generic situation as $\psi$ is varied, then the algebraic 2 translates into two distinct simple vectors in $Q$. However, there is a complex codimension 1 strata where there is a single degenerate point of geometric intersection, an isolated tangency between $S$ and $Q$. The situation is fully analogous to the complex equation $z^2 + a = 0$; the roots are distinct unless $a = 0$ in which case the root 0 also solves the derivative of the initial equation: $2z = 0$. It turns out that this behavior (generic, double) vs.\ (tangential, a single simple vector) correlates perfectly across $R_{BC}$, $R_{CA}$, and $R_{AB}$ and distinguishes $\hat{\mathrm{GHZ}}$ from $\hat{\mathrm{W}}$, the latter corresponding to single simple vectors. $\hat{\mathrm{W}}$ is characterized by having a unique simple vector in some (all) $R_{\ast\ast}$.

This characterization of $\hat{\mathrm{GHZ}}$ and $\hat{\mathrm{W}}$ in terms of simple vectors explains why $V_3 \subset \C P^7$ is an algebraic hypersurface. One can go further and build up a fairly precise topological picture of the stratification $V_1 \subset V_2 \subset V_3 \subset S^{15}$, but this will be deferred until it is necessary for an application. We turn next to a summary of the argument \cite{dvc00} that the behavior of the simple vectors in $R_{\ast\ast}$ in fact distinguishes the SLOCC classes $\hat{\mathrm{GHZ}}$ and $\hat{\mathrm{W}}$. In the lowest strata $V_1$ all the density matrices have rank 1. In the next strata $V_2 \setminus V_1$ two of the three density matrices have rank 2, e.g.\ for $A - BC$, $r_a = 1$, but $\rho_{CA}$ and $\rho_{AB}$ have ranges of the form $\C \otimes \C^2$.

The following lemma facilitates computation:

\begin{lemma}[\cite{dvc00}]
    If $\vert \mu \rangle \in H_E \otimes H_F$ is written $\sum_{i=1}^l \vert e_i \rangle \vert f_l \rangle$ then $R(\rho_E)$ lies in span$(\{\vert e_i\rangle\}_{i=1}^l)$.
\end{lemma}

\begin{proof}
    $\rho_E = \sum_{i,j=1}^l \langle f_i \mid f_j \rangle \lvert e_i \rangle \langle e_j \rvert$. For $\vert \nu \rangle$ in $R(\rho_E)$, let $\vert \nu \rangle = \rho_E \vert \mu \rangle$, some $\vert \mu \rangle$. Then $\vert \nu \rangle = \sum_{i,j=1}^l \langle f_i \mid f_j \rangle \langle e_i \mid \nu \rangle \vert e_j \rangle$.
\end{proof}

Suppose $\vert \psi \rangle$ is \emph{not} in $V_2^\pr$ but that $\rho_{BC}$ contains two independent simple vectors $\vert b_1 \rangle \vert c_1 \rangle$ and $\vert b_2 \rangle \vert c_2 \rangle$, then following \cite{dvc00} we may write
\begin{equation}\label{eq:rep}
    \vert \psi \rangle = \vert a_1 \rangle \vert b_1 \rangle \vert c_1 \rangle + \vert a_2 \rangle \vert b_2 \rangle \vert c_2 \rangle
\end{equation}
here $\vert a_i \rangle$, $i = 1,2$, is defined by $\vert a _i \rangle = \langle \xi_i \mid \psi \rangle$ for $\{\langle \xi_i \vert \}$ biorthogonal to $\{ \vert b_1 \rangle \vert c_1 \rangle, \vert b_2 \rangle \vert c_2 \rangle \}$, i.e. $\langle \xi_i \vert b_i \rangle \vert c_i \rangle = \delta_{ij}$.

This shows that in the generic case of two independent simple vectors in $R(\rho_{BC})$, $\vert \psi \rangle$ has tensor-rank 2. It is immediately that $\vert \psi \rangle$ is SLOCC equivalent to the canonical GHZ state.

It is shown in \cite{dvc00} that all vectors $\vert \psi \rangle$ with a degenerate simple vector in $\rho_{BC}$ are SLOCC equivalent to W. Such states, those of $\hat{\mathrm{W}}$, can be characterized as having tensor-rank 3, which is equivalent to $R(\rho_{BC})$ having dimension 2, with a unique simple (tensor-rank 1) vector.

With this background on the SLOCC classification, we now consider our next example.

\begin{example}\label{ex:qpk}
    Let $QP^k$, $k = 1$, 2, or 3, the quaternionic projective space, serve as our (4, 8, or 12)-dimensional base. Let $A$, $B$, and $C$ be copies of the canonical quaterion line bundle $Q \ra E \ra QP^2$, regarded as a $\C^2$-bundle; $Q \cong \C \oplus j \C$. These 2-plane bundles have $c_2 = 1 \in H^4(QP^k; \Z)$, and for $k=1$ we already met these bundles in Examples \ref{ex:s4} and \ref{ex:s4}\textprime. The bundles $A$, $B$, and $C$ have no nonvanishing sections. For dimensional reasons, $16 > \mathrm{base} \dim$, $P = A \otimes B \otimes C$ \emph{does} have nonvanishing sections, but no such section $\Gamma(x)$, $x \in QP^k$, can be of GHZ-type for all $x$.
\end{example}

Since the ket $\vert a_i \rangle$ is constructed by the formula (\ref{eq:rep}), and since the biorthogonals $\{\xi_i\}$, $i = 1,2$ are unique, any $\vert \psi \rangle$ in the SLOCC class $\hat{\mathrm{GHZ}}$ has a unique representation as in line \ref{eq:rep}, although the order of the two summands is not well defined. Indeed, by a fundamental theorem of Thom \cite{thom69}, the closure of $\hat{\mathrm{W}} \coloneqq V_3^\pr$ carries a fundamental top-dimensional\footnote{In general, $\Z_2$ coefficients are required to define the top class of a real algebraic variety, but for dimensional reasons the coefficients may be lifted to $\Z$ in this case.} cycle $\omega \in H_{13}(V^\pr;\Z)$, and if $\gamma \subset S^{15}$ is a linking circle to this cycle, if one chooses a ``first'' term in (\ref{eq:rep}) at some point of $\gamma$ and continuously propogates this choice along $\gamma$, the choice will be \emph{reversed} upon the first return. The degeneration of the two independent simple vector terms is modeled locally precisely as a quadratic branch point. In the special case that $\gamma$ is a small linking circle (to $V_3^\pr$) normal to $V_1^\pr$, the lowest strata, the two terms approach each other in the limit; in the general case their sum approaches a tensor of rank 3, i.e.\ an element of $\hat{\mathrm{W}}$. $\hat{\mathrm{W}}$ is said to have \emph{border rank} 2.

Now suppose, for a contradiction, that $\Gamma(x) \in \hat{\mathrm{GHZ}}$, for all $x \in QP^k$. Because $H_1(QP^k; \Z_2) \cong 0$, no 1-cycle $\Gamma(\alpha(\te))$, $\alpha: S^1 \ra QP^k$, can link (mod 2) $V_3^\pr \subset S^{15}$. So there would be no loop $\alpha(\te)$ in $QP^k$ along which the terms of the unique decomposition (\ref{eq:rep}) of $\vert \psi \rangle$ are exchanged. But this means that taking just the ``first'' term of (\ref{eq:rep}) defines a new section $\Gamma(x)$ of $P$ with tensor-rank 1, $\Gamma(x) = \vert a_1(x) \rangle \vert b_1(x) \rangle \vert c_1(x) \rangle$. But this is a contradiction as any of the three tensor factors, say $\vert a_1(x) \rangle$, would be a section of its bundle, $A$, at least modulo phase. However, since $H^2(QP^k;\Z) \cong 0$, there is no obstruction---as we saw in Example \ref{ex:s4}, to lifting a projective section of $\hat{A} \coloneqq A \slash \text{phase}$ to an actual section of $A$.

\begin{example}\label{ex:qp3}
    Let the base now be $QP^3$, and $A,B,C$ again be the canonical bundles we met in Example \ref{ex:qpk}. Now we make a claim \emph{dual} to the claim of Example \ref{ex:qpk}. We claim that $P = A \otimes B \otimes C$ cannot have an untwisted (nonvanishing) section $\Gamma(x)$, where $\Gamma(x) \in V_2^\pr$ for all $x$, $x \in QP^3$. That is, $\Gamma(x)$ cannot lie entirey in the lowest two strata of states.
\end{example}

\begin{proof}
    It is immediate, even without the untwisted hypothesis, that $\Gamma(x)$ cannot lie entirely in the very lowest strata $V_1^\pr$; this is how we finished off the analysis of Example \ref{ex:qpk}. If $\Gamma(x) = \vert a(x) \rangle \vert b(x) \rangle \vert c(x) \rangle$, then each of the three tensor factors implies, at first, a projective, and then a genuine section of the corresponding factor bundle, a contradiction. Because $H^2(QP^3, \Z) = 0$, the Chern obstructions must vanish, and $\Gamma(x)$ be untwisted.

    The idea for the general case is to write $QP^3$ as the union of three closed sets $QP^3 = X_{BC} \cup X_{CA} \cup X_{AB}$, where $X_{BC} = \Gamma^{-1}(A - BC)$, etc., overlapping along $\Gamma^{-1}(A - B - C)$. We will apply a bit of reasoning familiar in topology from the Listernick-Schirlemann (LS) theorem. But first, a technical point. The LS argument is cohomological and requires that the excision axiom apply to $X_1 \coloneqq X_{BC}$, $X_2 \coloneqq X_{CA}$, and $X_3 \coloneqq X_{AB}$. Excision may fail for certain pathological close sets like the ``$\sin \frac{1}{x}$-circle.'' However, the section $\Gamma(x)$ is into (a bundle of) real algebraic sets $V_2^\pr$, which obey favorable local conditions, \emph{Whitney stratification}, and therefore can be triangulated. Thus, the general section from a smooth base into fiber $V_2^\pr$ can be perturbed (without changing the range) so that the preimages of the strata are all subcomplexes of a Whitehead triangulation of the base, in this case $QP^3$. Thus w.l.o.g.\ we may apply excision to $X_1$, $X_2$, and $X_3$.

    Now suppose for a contradiction that the three inclusions induce the zero-map on $H^4(-;\Z)$:
    \begin{equation}
        H^4(QP^3;\Z) \xrightarrow{0} H^4(X_i;\Z) \text{ is zero},\ i = 1,2,3
    \end{equation}

    Then the generator $g \in H^4(\C P^3;\Z)$ pulls back to $g_i \in H^4(QP^3, X_i; \Z)$, $i = 1,2,3$, using the exact sequence of the pair $(QP^3,X_i)$. Now consider the commutative diagram of Figure \ref{fig:comm-diag}.

    \begin{figure}[H]
        \centering
        \begin{tikzpicture}[scale=1.2]
            \node at (0,1) {$H^4(QP^3, X_1; \Z) \otimes H^4(QP^3, X_2; \Z) \otimes H^4(QP^3,X_3;\Z) \xrightarrow{\cup} H^{12}(QP^3, \cup_{i=1}^3 X_i ;\Z) \cong 0$};
            \node[rotate=90] at (-4.8,0.6) {$\in$};
            \node at (-4.8,0.25) {$g_1$};
            \node[rotate=-90] at (-4.8,-0.1) {$\mapsto$};
            \node at (-4.8,-0.45) {$g$};
            \node[rotate=-90] at (-4.8,-0.75) {$\in$};
            \node at (-4.8,-1.1) {$H^4(QP^3;\Z)$};

            \node at (-3.55,-1.1) {$\otimes$};

            \node[rotate=90] at (-2.3,0.6) {$\in$};
            \node at (-2.3,0.25) {$g_2$};
            \node[rotate=-90] at (-2.3,-0.1) {$\mapsto$};
            \node at (-2.3,-0.45) {$g$};
            \node[rotate=-90] at (-2.3,-0.75) {$\in$};
            \node at (-2.3,-1.1) {$H^4(QP^3;\Z)$};

            \node at (-1,-1.1) {$\otimes$};

            \node[rotate=90] at (0.3,0.6) {$\in$};
            \node at (0.3,0.25) {$g_3$};
            \node[rotate=-90] at (0.3,-0.1) {$\mapsto$};
            \node at (0.3,-0.45) {$g$};
            \node[rotate=-90] at (0.3,-0.75) {$\in$};
            \node at (0.3,-1.1) {$H^4(QP^3;\Z)$};

            \node at (1.7,-1) {$\overset{\cup}{\longrightarrow}$};

            \node[rotate=90] at (3.1,0.6) {$\in$};
            \node at (3.1,0.25) {$0$};
            \node[rotate=-90] at (3.1,-0.1) {$\neq$};
            \node at (3.1,-0.45) {1};
            \node[rotate=-90] at (3.1,-0.75) {$\in$};
            \node at (3.4,-1.1) {$H^{12}(QP^3;\Z) \cong \Z$};

            \draw (0.6,0.15) -- (0.6,0.35);
            \draw[->] (0.6,0.25) -- (2.8,0.25);
            \draw (0.6,-0.55) -- (0.6,-0.35);
            \draw[->] (0.6,-0.45) -- (2.8,-0.45);
        \end{tikzpicture}
        \caption{}\label{fig:comm-diag}
    \end{figure}

    The contradiction $0 = 1$ in Figure \ref{fig:comm-diag} shows that for some $i = 1,2,3$, w.l.o.g.\ say $i = 1$, we must have $H^4(QP^3;\Z) \ra H^4(X_1;\Z)$ nonzero. Let $\lbar{g}$ denote the image of $g$ in $H^4(X_1;\Z)$. Tensoring with $Q$, the same argument shows $H^4(QP^3;Q) \ra H^4(X_1;Q)$ is an injection, and by hom-duality ($H^\ast(-;\text{Field})$ is naturally isomorphic to Hom$(H_\ast(-;\mathrm{Field}),\mathrm{Field})$) we have that:
    \begin{equation}
        H_4(X_1; Q) \ra H_4(QP^3; Q) \text{ is a surjection.}\label{eq:surj}
    \end{equation}

    Let\begin{tikzpicture}[baseline=-5.8ex]
        \node at (0,0) {$\C^8$};
        \draw [->] (0.4,0) -- (0.8,0);
        \node at (1.1,0) {$P_1$};
        \draw [->] (1.1,-0.3) -- (1.1,-0.7);
        \node at (1.1,-1) {$X_1$};
        \draw[->] (1.4,-0.9) to[out=60,in=-60] (1.4,-0.1);
        \node at (2.55,-0.5) {\footnotesize{$\vert a \rangle \vert \phi_{BC} \rangle = \Gamma_1$}};
    \end{tikzpicture}
    be the restriction of $P$ to $X_1$ and $\Gamma_1$ the restriction of $\Gamma$ to $X_1$. $P_1$ has a $(A - BC)$-type section $\vert a \rangle \otimes \vert \phi_{BC} \rangle$. Projectively splitting off the first factor, we obtain
    \begin{tikzpicture}[baseline=-5.8ex]
        \node at (0,0) {$\C P^1$};
        \draw [->] (0.4,0) -- (0.8,0);
        \node at (1.1,0) {$\hat{A}_1$};
        \draw [->] (1.05,-0.3) -- (1.05,-0.7);
        \node at (1.1,-1) {$X_1$};
        \draw[->] (1.4,-0.9) to[out=60,in=-60] (1.4,-0.1);
        \node at (1.8,-0.5) {\footnotesize{$\vert \hat{a} \rangle$}};
    \end{tikzpicture}
    , a projective line bundle with section $\vert \hat{a} \rangle$, obtained by projectivizing $A \ra \hat{A}$ and restricting to $X_1$.

    What can we do with this projective bundle? We do \emph{not} know that $H^2(X_1, Z) \cong 0$, so, in principle, there could be a first Chern class $c_1$ obstruction to lifting back to the vector bundle. However, our assumption that $\Gamma$, hence $\Gamma_1$, is untwisted w.r.t.\ $A \otimes (B \otimes C)$ allows us to lift $\vert \hat{a} \rangle$ to a section $\vert a \rangle$ of
    \begin{tikzpicture}[baseline=-5.8ex]
        \node at (0,0) {$\C^2$};
        \draw [->] (0.4,0) -- (0.8,0);
        \node at (1.1,0) {$A$};
        \draw [->] (1.1,-0.3) -- (1.1,-0.7);
        \node at (1.1,-1) {$X$};
    \end{tikzpicture}
    . But line \ref{eq:surj} implies $QP^1$ is rationally homologous to some 4-cycle $Y \subset X$, and $c_2(A)[QP^1] = 1$ by construction, implying $c_2(A)[Y] = c_2(A_1)[Y] \neq 0$. But we have just split a section $\vert a \rangle$ off $A$, so, stably $A_1$ is only a line bundle over $X_1$, thus by the dimension axiom $c_2(A_1) = 0$, a contradiction. This contradiction shows that $P$ has no untwisted section of type $V_2^\pr$; tripartite entanglement must arise at some point of the base.
\end{proof}

One may well ask what happens if we abandon the hypothesis that $\Gamma$ is untwisted; could there then be a section $\Gamma(x)$ with only bipartite entanglement? We do not know but the following example is cautionary: The first possibility to come to mind for a 4-cycle $Y$ as above would be $Y = QP^1 \cong S^4$. Since $H^2(S^4;\Z) \cong 0$ the section of $P_1$, at least over $Y \subset X$, would automatically be untwisted so that assumption would be redundant. However, suppose $Y \cong \C P^2$. This possible $\C P^2$ certainly embeds in the 12D normal bundle of $QP^1 \subset QP^3$, representing $[QP^1] \in H_4(QP^3; \Z)$. Now, most curiously the unique $\C^2$-bundle $T$ over $\C P^2$ with $c_2(T)[\C P^2] = 1$, is the sum of two line-bundles $T = T_+ \oplus T_-$, where the total Chern classes are:
\begin{align*}
    & c(T_+) = 1 + \mathrm{gen} \\
    & c(T_-) = 1 - \mathrm{gen} \\
    & c(T_+ \oplus T_-) = (1 + \mathrm{gen})(1 - \mathrm{gen}) = 1 - \mathrm{gen}^2
\end{align*}
gen generating $H^2(\C P^2;\Z)$ and via the choice over orientation, $-\mathrm{gen}^2$ being the orientation class in $H^4(\C P^2; \Z)$.

Thus if $Y = \C P^2$, $T$ would be $A \vert_Y$ the original bundle $A$ restricted to $Y$. Since $T$ is a sum of line bundles, $A \vert_Y$ would certainly have (two independent) projective sections. If $\vert \hat{a} \rangle$ were either of these projective sections, $c_A$, the first Chern obstruction would be a generator of $H^2(\C P^2; \Z)$, so $\vert \hat{a} \rangle$ would not be untwisted.

Perhaps starting with $X_1 = \operatorname{neighborhood}(\C P^2)$ and building $X_2$ and $X_3$ appropriately, a twisted but bipartite-entangled section $\Gamma$ of $P$ might be constructed. This is an attractive open problem.

We conclude this section with a final example, a variant of Example \ref{ex:qp3}, exhibbting the same phenomena but now with a more familiar base, the 12-torus $T^{12}$ rather than $QP^3$. Tori arise in condensed matter as Brillion Zones (``momentum tori'') and do not have the excotic flavor of quaterionic projective spaces.

\begin{customex}{4\textprime}
    Let $f: T^4 \coloneqq S^1 \times S^1 \times S^1 \times S^1 \ra S^4$ be a degree 1 map. Place over $S^4$ the $\C^2$-bundle
    \begin{tikzpicture}[baseline=-6.3ex]
        \node at (0,0) {$\C^2$};
        \draw [->] (0.4,0) -- (0.8,0);
        \node at (1.1,0) {$\Z$};
        \draw [->] (1.1,-0.3) -- (1.1,-0.7);
        \node at (1.1,-1) {$S^4$};
    \end{tikzpicture}
    already seen in Examples \ref{ex:qpk} and \ref{ex:qp3} with $c_2(\Z)[S^4] = 1$. Let $A$, $B$, and $C$ be three $\C^2$ bundles over $T^{12}$ obtained by projecting $T^{12}$ to $T^4$, composing with $f$ and then pulling back $\Z$. By definition, $A$ is obtained using $\pi_{1,\dots,4}: T^{12} \ra T^4$ with
    \[
        \pi_{1,\dots,4}(\te_i) = \begin{cases}
            \te & i = 1,2,3,4 \\
            0 & i > 4
        \end{cases}
    \]
    so $A = (f \circ \pi_{1,\dots,4})^\ast(\Z)$. Similarly, for $B$, replace $\pi_{1,\dots,4}$ with $\pi_{5,\dots,8}$ and for $C$ replace $\pi_{1,\dots,4}$ with $\pi_{9,\dots,12}$. Similar to example \ref{ex:qp3}, $A$, $B$, and $C$ have no nonvanishing section, since they each have a nonvanishing $c_2$. Unlike the previous case, these second Chern classes $c_A$, $c_B$, and $c_C$ are all distinct, but like the preceding case $c_A \cup c_B \cup c_B = \text{orientation class} \in H^{12}(T^{12};\Z)$. Let $c_A^\ast, c_B^\ast$, and $c_C^\ast$ be the hom-duals, lifted to the integers in $H_4(T^{12};\Z)$.
\end{customex}

\subsection*{Claim}
The bundle $P = A \otimes B \otimes C$ admits no untwisted (nonvanishing) section $\Gamma$ with values $\Gamma(x) \in V_2^\pr$, the variety of merely bipartite extanglement.

The proof is parallel to Example \ref{ex:qp3}. The Listernick-Schirlemann argument will find a class of essential 4-cycle $Y$, say $[Y] = c_A^\ast \in H_4(T^{12},\Z)$. As in Example \ref{ex:qp3}, the projective section $\langle \hat{a} \vert$ is lifted to $\langle a \vert$ over $Y$ (using the untwisted hypothesis) and then $\langle a \vert$ contradicts the second Chern class of $A$, $c_A \neq 0$, when restricted to $H^4(Y; \Z)$.

For completeness, the relevant Listernick-Schirlemann diagram is rendered below:

\begin{figure}[ht]
    \centering
    \begin{tikzpicture}[scale=1.1]
        \node at (0,1) {$H^4(T^{12}, X_A; Q) \otimes H^4(T^{12}, X_B; Q) \otimes H^4(T^{12},X_C;Q) \xrightarrow{\cup} H^{12}(T^{12}, T^{12} ; Q)$};
        \node[rotate=20] at (-5.8,0.6) {\footnotesize{$\lbar{g}_A \in$}};
        \node[rotate=20] at (-3,0.6) {\footnotesize{$\lbar{g}_B \in$}};
        \node[rotate=20] at (-0.3,0.6) {\footnotesize{$\lbar{g}_C \in$}};
        \node[rotate=20] at (2.7,0.6) {\footnotesize{$0 \in$}};
    
        \node[rotate=-90] at (-5.9,0.05) {$\mapsto$};
        \draw[->] (-4.6,0.7) -- (-4.6,0.3);
        \node at (-4.6,0) {$H^4(T^{12};Q)$};
        \node at (-3.45,0) {$\otimes$};
        \node[rotate=-90] at (-3.1,0.05) {$\mapsto$};
        \draw[->] (-1.8,0.7) -- (-1.8,0.3);
        \node at (-1.8,0) {$H^4(T^{12};Q)$};
        \node at (-0.7,0) {$\otimes$};
        \node[rotate=-90] at (-0.4,0.05) {$\mapsto$};
        \draw[->] (0.9,0.7) -- (0.9,0.3);
        \node at (0.9,0) {$H^4(T^{12};Q)$};
        \node at (2.15,0.1) {$\xrightarrow{\cup}$};
        \node[rotate=-90] at (2.6,0.1) {$\mapsto$};
        \draw[->] (3.95,0.7) -- (3.95,0.3);
        \node at (3.95,0) {$H^4(T^{12};Q)$};
    
        \node[rotate=20] at (-5.8,-0.3) {\footnotesize{$g_A \in$}};
        \node[rotate=20] at (-3,-0.3) {\footnotesize{$g_B \in$}};
        \node[rotate=20] at (-0.3,-0.3) {\footnotesize{$g_C \in$}};
        \node[rotate=20] at (2.75,-0.3) {\footnotesize{$1 \in$}};

        \node[rotate=-90] at (-5.9,-0.85) {$\mapsto$};
        \node[rotate=20] at (-5.8,-1.25) {\footnotesize{$0 \in$}};
        \draw[->] (-4.6,-0.3) -- (-4.6,-0.7);
        \node at (-4.7,-1) {$H^4(X_A;Q)$};
        \node[rotate=-90] at (-3.1,-0.85) {$\mapsto$};
        \node[rotate=20] at (-3,-1.25) {\footnotesize{$0 \in$}};
        \draw[->] (-1.8,-0.3) -- (-1.8,-0.7);
        \node at (-1.9,-1) {$H^4(X_B;Q)$};
        \node[rotate=-90] at (-0.4,-0.85) {$\mapsto$};
        \node[rotate=20] at (-0.3,-1.25) {\footnotesize{$0 \in$}};
        \draw[->] (0.9,-0.3) -- (0.9,-0.7);
        \node at (0.8,-1) {$H^4(X_C;Q)$};
    \end{tikzpicture}
    \caption{}
\end{figure}

Let $g_A$ ($g_B,g_C$) be $\pi_{1,\dots,4} \circ f[S^4](\pi_{5,\dots,8} \circ f[S^4],\pi_{9,\dots,12}\circ f[S^4])$. If all three, $g_A$, $g_B$, and $g_C$ map to zero on the bottom row, then the lifts $\lbar{g}_A$, $\lbar{g}_B$, $\lbar{g}_C$ are defined and the same $0 = 1$ contradiction from Example \ref{ex:qp3} is again obtained. So in one case, say for $X_A$, it must be that $g_A$ maps not trivially into $H^4(X_A;Q)$. Then dually, $X_A$ contains a rational 4-cycle $Y$ carrying the class $c_A^\ast$. The lifted section $\vert a \rangle$ over $Y$ contradicts $c_2[A][Y] \neq 0$. \qed

Again, we do not know if the no-twist hypothesis can be removed.

\begin{example}[Families of $\hat{\mathrm{GHZ}}$ states from the Borromean rings]\label{ex:ghz}
    It has been noticed \cite{ara97} that the GHZ behaves under partial trace the same way the Borromean rings (from topology) behave under cutting a component. In both cases, everything falls apart; there is no residual entanglement; there is no residual linking.
\end{example}

The purpose of this extended example is to make this analogy precise by showing:
\begin{thm}
    For any $\mathrm{rank} = 2$ TQFT $V$, the state $\psi \in V(T_1^2 \perp\!\!\!\!\perp T_2^2 \perp\!\!\!\!\perp T_3^2) \cong (\C^2)^{\otimes 3}$, determined by the Borromean ring complement, lies in $\hat{\mathrm{GHZ}}$, $\psi \in \hat{\mathrm{GHZ}}$.
\end{thm}

\begin{proof}
    The quantum dimension of $V$ is $\mathcal{D} = \sqrt{\sum d_i^2}$, $d_i$ the dimensions of the various particle types. Of course in the rank 2 case, there are only two particle types. Set $\delta = \frac{1}{\mathcal{D}}$. For any TQFT \cite{walker91} the partition function $Z$ of the following simple manifold may be expressed as:
    \begin{equation}
        Z(S^3) = \delta,\ Z(S^1 \times S^2) = 1,\ Z(S^1 \times S^2\ \#\ S^1 \times S^2) = \delta, \text{ and } Z(T^3) = \operatorname{rank}(V) = 2
    \end{equation}

    There are $2^3 = 8$ closed 3-manifolds obtained by filling the components of the Borromean rings, where each component is filled to kill either the meridian or the longitude. Precisely, the possibilities listed above arise. Filling so as to kill the meridian simply eliminates that component. So, $S^3$ results from three $m$-fillings. The 3 way of having 2 meridional filling yield $S^1 \times S^2$ (which is 0-framed surgery on the unknot). The three 1-meridional filling yields 0-framed surgery on the 2-component unlink, $S^1 \times S^2\ \#\ S^1 \times S^2$. The 8th possibility, three longitude filling, is 0-framed surgery on the Borromean rings which (see \cite{thurston80}) is $T^3$. These 8 partition functions are the coordinates of $\psi$ in the non-orthogonal basis where the projector to the trivial particle, $\omega_0$, is located as meridian or longitude on the three tori of the Borromean rings. Topologically, the projection $\omega_0$ signifies filling so that any simple closed curve labeled by $\omega_0$ bounds a disk (or, thickening, a 2-handle). For a given torus $T^2$, if we use the basis $\{\vert 0 \rangle, \vert 1 \rangle\}$ for $V(T^2)$, $\vert 0 \rangle = \omega_0$ on meridian and $\vert 1 \rangle = \omega_0$ on longitude, we compute that $\langle 0 \vert 0 \rangle = \langle 1 \vert 1 \rangle = Z(S^1 \times S^2) = 1$ and $\langle 0 \vert 1 \rangle = \langle 1 \vert 0 \rangle = Z(S^3) = \delta$, $0 < \delta < 1$, so indeed this basis is \emph{not} orthonormal.

    However, from the definition of SLOCC (line \ref{eq:slocc}), this classification is $\operatorname{GL}(2\,C) \times \operatorname{GL}(2,\C) \times \operatorname{GL}(2,\C)$-invariant. So we may use this basis to determine the SLOCC-class of $\psi$. To do this, compute $\rho_{BC} = \operatorname{tr}_A(\vert \psi \rangle \langle \psi \vert)$. The result is:
    \begin{equation}
        \def\arraystretch{1.2}
        \rho_{BC} =
        \begin{array}{lllll}
            & \vert 0 \rangle \vert 0 \rangle & \vert 0 \rangle \vert 1 \rangle & \vert 1 \rangle \vert 0 \rangle & \vert 1 \rangle \vert 1 \rangle \\
            \langle 0 \vert \langle 0 \vert & \delta^2 + 1 & 2\delta & 2 \delta & \delta^2 + 2 \\
            \langle 0 \vert \langle 1 \vert & 2\delta & \delta^2 + 1 & \delta^2 + 1 & 3\delta \\
            \langle 1 \vert \langle 0 \vert & 2\delta & \delta^2 + 1 & \delta^2 + 1 & 3\delta \\
            \langle 1 \vert \langle 1 \vert & \delta^2+2 & 3\delta & 3\delta & \delta^2+4 \\
            & 
        \end{array}
        \vspace{-1em}
    \end{equation}

    The sum of the first two columns divided by $(\delta+1) = \begin{pmatrix}
        \delta + 1 \\ \delta + 1 \\ \delta+1 \\ \delta+2
    \end{pmatrix}$, subtracting this from the first and dividing by $(\delta-1)$ yields $\begin{pmatrix}
        \delta \\ 1 \\ 1 \\ \delta
    \end{pmatrix}$, so taking the difference of the two, we see range$(\rho_{BC}) =: R_{BC}$ is spanned by $\begin{pmatrix}
        \delta \\ 1 \\ 1 \\ \delta
    \end{pmatrix}$ and  $\begin{pmatrix}
        1 \\ \delta \\ \delta \\ 2
    \end{pmatrix}$.

    Now let us work out whether there are one or two simple vectors (up to a scalar) in $R_{BC}$. The general element of $R_{BC}$ has the form:
    \begin{equation}\label{eq:matrix-form}
        x\begin{pmatrix}
            \delta \\ 1 \\ 1 \\ \delta
        \end{pmatrix} + y\begin{pmatrix}
            1 \\ \delta \\ \delta \\ 2
        \end{pmatrix} = \begin{pmatrix}
            \delta x + y \\ \delta y + x \\ \delta y + x \\ \delta x + 2y
        \end{pmatrix}
        \begin{array}{l}
            \vert 0 \rangle \vert 0 \rangle \\ \vert 0 \rangle \vert 1 \rangle \\ \vert 1 \rangle \vert 0 \rangle \\ \vert 1 \rangle \vert 1 \rangle
        \end{array}
    \end{equation}
    Again, up to a non-zero scalar, simple vectors $\te$ will have one of four forms: $\vert 1 \rangle \otimes \vert 1 \rangle$, $(\vert 0 \rangle + p\vert 1 \rangle) \otimes (q \vert 1 \rangle)$, $(p\vert 1 \rangle) \otimes (\vert 0 \rangle + q \vert 1 \rangle)$, or $(\vert 0 \rangle + p\vert 1 \rangle) \otimes (\vert 0 \rangle + q \vert 1 \rangle)$. The first three cases can be eliminated immediately as not being of form (\ref{eq:matrix-form}). The 4th possibility for a simple vector expands to:
    \begin{equation}
        \te = 1 \vert 0 \rangle \vert 0 \rangle + q \vert 0 \rangle \vert 1 \rangle + p \vert 1 \rangle \vert 0 \rangle + pq \vert 1 \rangle \vert 1 \rangle
    \end{equation}
    Matching to (\ref{eq:matrix-form}), we find: $\delta x + y = 1$, $\delta y + x = q = p$, and $\delta x + 2y = p^2$, so $y = p^2 - 1$, and $\delta y + x = p \implies \delta p^2 - \delta + x = p$, or $x = \delta - \delta p^2 + p$.

    We have a final equation to use: $p^2 = \delta x + 2y$, or $p^2 = \delta(\delta - \delta p^2 + p) + 2(p^2 - 1)$, or
    \begin{equation}
        (1-\delta^2)p^2 + \delta p + \delta^2 - 2 = 0
    \end{equation}

    With $\delta = \frac{1}{\mathcal{D}}$, fixed by the choice of TQFT $V$, this quadratic equation for $p$ has a double root iff
    \begin{equation}\label{eq:double-root}
        \delta^2 - 4(1-\delta^2)(\delta^2-2) = 0
    \end{equation}
    Setting $\alpha = \delta^2$, (\ref{eq:double-root}) implies the quadratic equation
    \begin{align}
        & 4 \alpha^2 - 11 \alpha + 8 = 0, \text{ or } \\
        & \alpha = \frac{11 \pm \sqrt{121-128}}{8}
    \end{align}

    Since $\alpha$ cannot be real and $\delta$ must be, this proves that for all 2D TQFT $V$, $\rho_{BC}$ has two distinct simple vectors (up to scale) and so $\psi \in \hat{\mathrm{GHZ}}$.
\end{proof}

To weave Example \ref{ex:ghz} into the bundle context of the previous examples, note the $\Z_3$, 120 degree rotational symmetry of the Borromean rings. Setting $B = $ Borromean rings complement and $S^\infty$ the colimit of the finite, odd-dimensional spheres under inclusion, consider $(B \times S^\infty) \slash \Delta =: E$, where $\Delta$ is the diagonal action: $120^\circ$ rotation on $B$, and multiplication by the phase $e^{2\pi i/3}$ or $S^\infty$. $E$ is a flat bundle over $S^3 \slash e^{2\pi i/3} = L_3^\infty$, the infinite dimensional dense space with $\pi_1(L_3^\infty) \cong \Z_3$, with fiber B.
\begin{equation}
    \begin{tikzpicture}[baseline=(current bounding box.center)]
        \node at (-2,0.2) {$T^2 \perp\!\!\!\!\perp T^2 \perp\!\!\!\!\perp T^2 \xrightarrow{\de} B \ra E$};
        \draw[->] (0,-0.3) -- (0,-0.8);
        \node at (0,-1.2) {$L_3^\infty$};
        \draw[->] (0.5,-0.5) -- (2.5,-0.5);
        \node at (1.5,-0.75) {\footnotesize{apply $V$}};
        \node at (4,0.1) {$(\C^2)^{\otimes 3} \ra \mathbb{V}$};
        \draw[->] (4.9,-0.3) -- (4.9,-0.8);
        \node at (4.95,-1.2) {$L_3^\infty$};
        \draw[->] (4.6,-1) -- (3.4,-0.3);
        \node at (3.6,-0.9) {\footnotesize{$\psi(B)$}};
    \end{tikzpicture}
\end{equation}

Now we may apply the functor $V$ fiber-wise to produce a $(\C^2)^{\otimes 3}$-bundle $\mathbb{V}$ over $L_3^\infty$. $\mathbb{V}$ is interesting; it is locally the tensor product of three $\C^2$-bundles, but not globally---the factors get permuted by the cyclic symmetry implemented by $\pi_1(L_3^\infty) \cong \Z_3$. However, the local tensor structure is all that we need to discuss the SLOCC-type of the section $\psi(B)$ that the bulk of the fiber induces on $V(\de B) = V(T^2 \perp\!\!\!\!\perp T^2 \perp\!\!\!\!\perp T^2) \cong (\C^2)^{\otimes 3}$. As we have just seen, the section $\psi(B)$ is of type $\hat{\mathrm{GHZ}}$, over all points of the base $L_3^\infty$.

\section{Further possibilities from quantum topology}
In quantum topology the Hilbert space $H(Y)$ of a high genus $(g)$ surface $Y$ is assembled from elementary pieces via a direct sum over charge sectors of tensor products. Entanglement is also a natural concept in this context. Furthermore, these $H(Y)$ naturally bundle over the moduli space $M(g)$ of conformal structures on $Y$. These bundles have the Knizhnik-Zamolodchikov flat connection whose holonomy is the Jones representation, with nonvanishing parallel sections corresponding to fixed vectors. For TQFTs with irreducible representations, the entanglement properties of non-parallel sections are of interest. When the quantum dimension of the theory is sufficiently large, $\dim(H(Y_g)) >> \dim(M(g))$, $g$ large, so there is an abundance to nonvanishing sections to study.

The sum over charge sectors may be formally accomplished by working locally with graded Hilbert space, graded by charge labels on all boundary components of the surface. The concept of entanglement depends on a decomposition. One might first consider what happens when a surface is decomposed into the smallest possible subsurfaces. Any surface may be decomposed into ``pairs of pants'' (sometimes called ``tricons''), and for many TQFTs, e.g.\ all the $\operatorname{SU}(2)_k$ theories, vacuum states (in the purely topological sector) will \emph{all} be unentangled w.r.t.\ a pants decomposition, since for fixed charge labels the Hilbert space $H(\mathrm{pants})$ always has dimension 1. This can equivalently be expressed as: ``fusion channels are unique.'' Thus to study topological ground state entanglement of pants decompositions (in a purely topological context as opposed to the hybrid context of \cite{kp06}) we will need to consider theories, such as $\operatorname{SU}(3)_3$, with non-trivial channel multiplicities.

\section{Some thoughts on applications}
We present some initial thoughts on applications of these ideas in physics and quantum information.

Consider a two electrons moving in some periodic $d$-dimensional crystal lattice.  Following \cite{kit09}, in the absence of interaction we can describe each band by some (possibly nontrivial) bundle over the momentum torus $\mathbb{T}^d$.
These bundles have fiber $\mathbb{C}^n$ for some $n$, where the fiber describes some internal degrees of freedom (e.g., spin or valley degrees of freedom), depending on the particular crystal.
Now suppose we add some weak interaction between these electrons, sufficiently weak that the electrons each remain in the given band.  If the interaction is translationally invariant, the total crystal momentum (i.e., the sum of coordinates in the two momentum tori, modulo $2\pi$) is conserved.  Consider wavefunctions with given total crystal momentum (e.g. total crystal momentum $0$), any such wavefunction is then a section in a product bundle: if the total crystal momentum is zero, then the coordinate in one base space is minus the coordinate in the other, so the base is still $\mathbb{T}^d$.  Then, in some cases we can say that any such wavefunction either is entangled in the internal degrees of freedom or has some zero somewhere.

This particular method of implementing product bundles 
gives something more general in the case of a product of three or more bundles.  If we had instead three electrons, then only the total momentum is conserved, and so the base is now $\mathbb{T}^{2d}$.
Indeed, the usual tensor product of two bundles is an outer tensor product pulled back by a diagonal inclusion, and in this case of three electrons we have an outer product of three bundles pulled back to a $\mathbb{T}^{2d}$ in $\mathbb{T}^{3d}$ where the coordinates add to zero in triples.
However, we may imagine implementing higher product bundles by considering a system dependent on some parameters.  If we have several such systems, dependent on some parameters, we have a product bundle, and if then we allow those parameters themselves to become quantum degrees of freedom, again a wavefunction defines a section and our results may imply either vanishing of the wavefunction somewhere or entanglement.

Another speculative application is to quantum codes.  It has been argued\cite{gottesman2013fibre} that quantum codes can be described by section of a fiber bundle, where the base corresponds to a choice of stabilizers of the code and the fiber describes the encoded logical information.  It was conjectured that fault tolerant operations correspond to flat connections on this bundle.  Then, a product bundle naturally appears by considering several copies of the same code.  A nowhere-vanishing section is, up to phase, a projector that projects onto a given state of the logical operators.  Thus, in some cases it may be possible to say something about the complexity of certain logical operations implemented in a fault tolerant way on these quantum codes.

\bibliography{references}

\begin{bibdiv}
\begin{biblist}

\bib{ara97}{incollection}{
      author={Aravind, P.K.},
       title={Borromean entanglement of the {GHZ} state},
        date={1997},
   booktitle={{Potentiality, Entanglement and Passion-at-a-Distance}},
      series={Boston Studies in the Philosophy of Science},
      volume={194},
   publisher={Springer},
       pages={53\ndash 59},
}

\bib{bs18}{article}{
      author={Brown, Adam},
      author={Susskind, Leonard},
       title={Second law of quantum complexity},
        date={2018},
     journal={Phys. Rev. D},
      volume={97},
       pages={086015},
}

\bib{dvc00}{article}{
      author={D{\"u}r, Wolfgang},
      author={Vidal, Guifre},
      author={Cirac, Ignacio},
       title={Three qubits can be entangled in two inequivalent ways},
        date={2000},
     journal={Phys. Rev. A},
      volume={62},
      number={6},
       pages={062314},
}

\bib{gottesman2013fibre}{article}{
      author={Gottesman, Daniel},
      author={Zhang, Lucy~Liuxuan},
       title={Fibre bundle framework for unitary quantum fault tolerance},
        date={2013},
     journal={arXiv preprint arXiv:1309.7062},
}

\bib{hastings13}{article}{
      author={Hastings, Matthew},
       title={Classifying quantum phases with the {K}irby torus trick},
        date={2013},
     journal={Phys. Rev. B},
      volume={88},
       pages={165114},
}

\bib{kit09}{inproceedings}{
      author={Kitaev, Alexei},
       title={Periodic table for topological insulators and superconductors},
        date={2009},
   booktitle={{AIP Conference Proceedings}},
      volume={1134},
       pages={22\ndash 30},
}

\bib{kp06}{article}{
      author={Kitaev, Alexei},
      author={Preskill, John},
       title={Topological entanglement entropy},
        date={2006},
     journal={Phys. Rev. Lett.},
      volume={96},
      number={11},
       pages={110404},
}

\bib{lp99}{article}{
      author={Lo, Hoi-Kwong},
      author={Popescu, Sandu},
       title={Classical communication cost of entanglement manipulation: {I}s
  entanglement an interconvertible resource?},
        date={1999},
     journal={Phys. Rev. Lett.},
      volume={83},
      number={7},
       pages={1459},
}

\bib{thom69}{article}{
      author={Thom, Ren{\'e}},
       title={Ensembles et morphismes stratifi{\'e}s},
        date={1969},
     journal={Bull. Amer. Math. Soc.},
      volume={75},
       pages={240\ndash 284},
}

\bib{thurston80}{misc}{
      author={Thurston, William~P.},
       title={The geometry and topology of three-manifolds},
        date={1980},
        note={Available at \texttt{http://library.msri.org/books/gt3m/}},
}

\bib{v99}{article}{
      author={Vidal, Guifr{\'e}},
       title={Entanglement of pure states for a single copy},
        date={1999},
     journal={Phys. Rev. Lett.},
      volume={83},
      number={5},
       pages={1046},
}

\bib{walker91}{misc}{
      author={Walker, Kevin},
       title={On {W}itten's 3-manifold invariants},
        date={1991},
        note={Available at
  \texttt{https://canyon23.net/math/1991TQFTNotes.pdf}},
}

\end{biblist}
\end{bibdiv}

\end{document}